\documentclass[10pt,journal,cspaper,compsoc]{IEEEtran}

\ifCLASSOPTIONcompsoc
\else
 \usepackage{cite}
\fi

\ifCLASSINFOpdf
\else
\fi

\ifCLASSOPTIONcompsoc
  \usepackage[caption=false,font=normalsize,labelfont=sf,textfont=sf]{subfig}
\else
  \usepackage[caption=false,font=footnotesize]{subfig}
\fi

\usepackage{graphicx}
\usepackage[caption=false]{subfig}
\usepackage{amsmath,amssymb}
\usepackage{amsthm}
\usepackage{multirow}
\usepackage{tabulary}

\newcommand{\INKA}{\texttt{INKA}}

\newtheorem{lemma}{Lemma}

\begin{document}

\title{INKA: An Ink-based Model of Graph Visualization}

\author{Quan~Hoang~Nguyen
\IEEEcompsocitemizethanks{\IEEEcompsocthanksitem Q. Nguyen is with the School of Information Technologies, University of Sydney, Australia.\protect\\
E-mail: \{quan.nguyen\}@sydney.edu.au
}
\thanks{}}

\markboth{Journal of \LaTeX\ Class Files,~Vol.~1, No.~1, January~2018}%
{Shell \MakeLowercase{\textit{et al.}}: Ink-based model of graph visualizations}


\IEEEcompsoctitleabstractindextext{%
\begin{abstract}
Common quality metrics of graph drawing have been about the readability criteria, such as small number of edge crossings, small drawing area and small total edge length. Bold graph drawing considers more realistic drawings consisting of vertices as disks of some radius and edges as rectangles of some width. However, the relationship that links these readability criteria with the rendering criteria in node-link diagrams has still not been well-established. 

This paper introduces a model, so-called \INKA~(Ink-Active), that encapsulates mathematically the relationship between all common drawing factors. Consequently, we investigate our \INKA~model on several common drawing algorithms and real-world graphs.
\end{abstract}

}

\maketitle

\IEEEdisplaynotcompsoctitleabstractindextext

%
\IEEEpeerreviewmaketitle


\ifCLASSOPTIONcompsoc
  \noindent\raisebox{2\baselineskip}[0pt][0pt]%
  {\parbox{\columnwidth}{\section{Introduction}\label{sec:introduction}%
  \global\everypar=\everypar}}%
  \vspace{-1\baselineskip}\vspace{-\parskip}\par
\else
  \section{Introduction}\label{sec:introduction}\par
\fi

%
%
\IEEEPARstart{D}{rawing} graphs has been extensively studied~\cite{Battista:1998:GDA:551884,DBLP:reference/crc/2013gd} and have been successfully applied in cross domains including social, business, and biology.


Readability criteria are the most common measures for the quality of graph drawings and have been aimed in many layout algorithms \cite{Battista:1998:GDA:551884,purchase1996validating,purchase1997aesthetic,DBLP:reference/crc/2013gd}. They include, for example, 
\begin{itemize}
	\item \emph{edge crossing}(few edge crossings), 
	\item \emph{edge length} (small total edge length), 
	\item \emph{area} (small area of a grid drawing).
\end{itemize}
These readability criteria are shown effective and algorithms built-in with these criteria have been successfully applied in many domains.
There have been a number of empirical evaluation of readability measures for graph drawing~\cite{purchase1996validating,purchase1997aesthetic,vismara2000experimental}.
Previous work has shown that improving multiple aesthetics can produce better graph drawing in terms of human perception~\cite{huang2013improving}.

A more realistic view of graph drawing considers rendering factors, such as \emph{node size} and \emph{edge width}~\cite{vanKreveld2011}. 
Bold graph drawing draws every vertex of a graph by a disk of radius $r$,
and every edge by a rectangle of width $w$, for some non-negative numbers $r$ and $w$~\cite{vanKreveld2011,Pach2012}.

Here, we distinguish drawing factors into two categories: \emph{layout} factors and  \emph{rendering} factors. The layout factors include, for examples, edge crossing, edge length and drawing area. The rendering factors include, for example, node size, edge width, node / edge colors, shading and transparency.

Despite a plethora of studies of the abstract graph readability concepts and the rendering of graphs as node-link diagrams, previous work has focused only on each quality type. There is not much research that has studied the relationship between these two types of quality.
As such,  optimized drawings for abstract graph readability do not consider actual rendered results; while rendering criteria alone are not sufficient with the absence of some readability criteria.

In this paper, we are interested in studying the relationship between the two types of drawing quality (the abstract graph readability and the rendering of the graphs) in node-link diagrams. In particular, we will investigate a relationship of the most common layout factors (edge length, edge crossings and drawing area) with the most common rendering factors (node size and edge width).
These drawing factors are crucial for a good graph drawing.




Specifically, we introduce a model, so-called \INKA~(Ink-Active), which mathematically addresses the relationship between the common drawing factors. We explore the relationship with the two criteria using the expression of the amount of ink. 
Intuitively, the new model has been built based on the well-known concept of ink-data ratio by Tufte~\cite{tufte1983visual}. 

With the \INKA~model, our aim is to determine the relationship for selected common drawing factors in graph drawing.
The \INKA~model also leads to some guidelines for choosing drawing factors in graph layout algorithms.

We must stress that we do not aim in this paper for an algorithm to optimize for both quality types, nor to compare between the importance of these criteria. Instead, these will be considered in our future work.




In summary, the paper makes the following contributions:
\begin{itemize}\setlength{\itemsep}{1pt}  \setlength{\parskip}{0pt}\setlength{\parsep}{0pt}
	\item We have proposed a new model, called \INKA, for expressing the relationship between the important layout factors (edge  crossings, edge length and the drawing area) with the rendering factors (node size, edge width).
	\item We examine our \INKA~model and evaluate it using several common graph drawing algorithms.
	\item We also evaluate \INKA~model using real-world graphs and standard multi-level force-directed graph layouts.
\end{itemize}

The rest of the paper is organized as follows. 
Section~\ref{sec:related} gives related work. Section~\ref{sec:inka} describes our \INKA~model. Section~\ref{sec:opti} gives some examples of optimization of rendered results using \INKA~model; Section~\ref{sec:studies} gives several studies of \INKA~on common drawing approaches. Section~\ref{sec:eval} gives some evaluation of the \INKA~model using real-world graphs. Section~\ref{sec:diskussion} gives some diskussions and Section~\ref{sec:conclusion} concludes.

\section{Related work~\label{sec:related}}



\subsection{Graph drawing quality metrics}
Criteria for `good' graph visualization have been investigated extensively~\cite{Battista:1998:GDA:551884}.
Graph drawing algorithms over the years typically take into account one or more aesthetic criteria for better \emph{readability} of the drawing.
These aesthetic criteria include, for example,
\begin {enumerate} 
\item minimizing the number of edge crossings~\cite{reingold1981tidier};
\item minimizing the total area~\cite{tamassia1988automatic};
\item edge lengths should be short but not too short~\cite{coleman1996aesthetics}.
\end{enumerate}
Amongst these aesthetics, small number of edge crossings is one of the most common criterion from previous user studies~\cite{purchase1997aesthetic}. Besides, the amount of ink and minimum total edge length has been used in many layout algorithms; for example, ~\cite{sugiyama1994methods,north2001online,gansner2006improved,gansner2011multilevel}. Achieving small total area is another common approach~\cite{tamassia1988automatic}. Overall, improving multiple aesthetics can produce better graph drawings~\cite{huang2013improving}.


However, there is not much research to model and understand the relationship between the drawing factors. 


\subsection{Ink model and data-ink}
Tufte's principle of 'maximizing \emph{data-ink}' is well-known for data visualization~\cite{tufte1983visual}. The \emph{data-ink} is the non-erasable ink that presents data; removing a data-ink from the display would cause information loss. The \emph{data-ink ratio} measures the ratio of data-ink to the total ink used. For a fixed piece of information, maximizing the ratio corresponds to minimizing the amount of ink.


A number of graph drawing algorithms have aimed for a minimum total edge length, or more precisely, a minimum amount of ink. This criterion has been  studied~\cite{sugiyama1994methods,north2001online,gansner2006improved,gansner2011multilevel}.

In this paper, we denote the concept of 
\emph{ink effectiveness}, which is the inverse of \emph{data-ink ratio}. This is equivalent to the ratio of ink over data. Given the pictures to visualize the same data (graph), the picture using less ink is more ink-effective.

\section{\INKA: Ink-Active model~\label{sec:inka}}

This section presents our \INKA~model that formally shows relationship between the abstract graph readability and the rendering of graphs. Specifically, the model aims to draw a connection between the most common drawing factors of the two types. The new model is based on the amount of ink used in drawing graphs.

\subsection{Problem definitions and notations}
We first define several criteria for 'proper' drawings of graphs. Some notations are borrowed from bold graph drawing~\cite{vanKreveld2011}.

Given a graph $G=(V,E)$, a layout algorithm decides a mapping of each node $v$ in $V$ to a location $p_v$ in 2D. Vertices then are drawn as solid disks of a radius $r$ and edges are represented by straight-line segments to connect adjacent nodes.
Edges are often considered as having zero or negligible width, but realistically they are drawn by rectangles with a width $w$. The values of disk radius $r$ and edge width $w$ are non-negative. A bold drawing $D$ of $G$ is the union of these disks and rectangles.

In this paper, a bold drawing is \emph{proper} if the following conditions are met:
\begin{enumerate}
	\item No two disks intersect.
	\item Any point in the drawing belongs to at most two edges.
	\item
	Any pair of edges can cross each other at most once.
\end{enumerate}
Throughout this paper, we only consider proper bold graph drawing, those that satisfy the above conditions.

Next, we present the \INKA~model, which  models mathematically the ink requirements from layout specifications.

\subsection{INKA-total\label{sec:inkaEq}}

The total ink used in a bold drawing $D$ of $G$ is given in the following \emph{INKA-total} equation:
\begin{equation}\label{eq:inka}
ink(D) = ink(V) + ink(E) - overlap,
\end{equation}

where 
\begin{itemize}
	\item $ink(V)$ is the total ink for all vertices; 
	\item $ink(E)$ is the total ink for all edges (minus the intersection between disks and rectangles); and 
	\item $overlap$ is the total ink that is saved from overlapping between the edges. 
\end{itemize}


Additional notations are given as follows: 
\begin{itemize}
	\item 
Let $l_e$ denote the length of an edge $e$.
\item 
Let $L$ be the sum of all edge length $L = \sum_{e \in E} l_e$. 
\item Let $cr(D)$ denote the number of edge crossing in $D$. 
\item Let $m$ ($n$) denote the total number of edges (nodes).
\end{itemize}

From the INKA-total equation in Equation~\ref{eq:inka}, the total ink in the drawing $D$ can be computed as follows:
\begin{lemma}
	\begin{equation}\label{eq:inka1}
	ink(D) = n \pi r^2 + w (L - 2mr) - w^2. cr(D)
	\end{equation}
\end{lemma}
\begin{proof}
	Each disk representing a vertex $v$ takes $\pi r^2$ pixels. The total amount of ink for the disks are $ink(V) = n r^2$.
	
	Each rectangle representing an edge $e$ has  length of $l_e - 2r$ and width of $w$. After subtracting the intersections of the two adjacent disks (vertices), the rectangle has length of $l_e-2r$; thus, the remaining rectangle takes $w(l_e - 2r)$ pixels.
	The total amount of ink for the disks is $ink(E) = \sum_{e \in E} w(l_e - 2r) = w (L - 2mr)$.
	
	The value of $overlap$ is equal to the total amount of ink that is saved from rectangle-rectangle crossing. That is, $overlap$ is proportional to the total number of edge crossings $cr(D)$ in the drawing $D$. Thus, $overlap$ can be approximated by $w^2 . cr(D)$.
	
	Thus, the total amount of ink in $D$ can be written precisely as $ink(D) = ink(V) + ink(E) - overlap = n \pi r^2 +  w(L - 2mr)  - w^2 .cr(D)$. \qed
\end{proof}	

Note that, the above approximation of ink is quite simplistic. However, this can give a quick estimation of the amount of ink for a graph layout before actual rendering of the graph.


\subsection{INKA-area\label{sec:inkaIneq}}
Here, we present a mathematical model of the relationships for \emph{drawing area} and \emph{drawing density}.

The \emph{drawing density} defines the proportion of total ink $ink(D)$ over total drawing area.
Intuitively, this approach aims for a drawing that requires small drawing area but at the same time the drawing is not too dense. In fact, small drawing area is a common criterion in graph drawing~\cite{tamassia1988automatic}.

	Then the drawing density satisfies the condition
$$ink(D) / A \leq \gamma,$$ where
\begin{itemize}
	\item $A$ is the drawing area. 
	\item $\gamma$ is the maximum drawing density that is good for drawing. 
\end{itemize}
 
 Generally, the value of $\gamma$ is specified by users. The default value of $\gamma$ is 1. 
 
The below \emph{INKA-area} inequality captures the relationship of the most common drawing factors, given by:
\begin{lemma}
	\begin{equation}\label{eq:drawingratio}
	ink(D) =  n \pi r^2 + w (L - 2mr) - w^2. cr(D) \leq \gamma A,
	\end{equation}
\end{lemma}
\begin{proof}
	From $ink(D) / A \leq \gamma$, we can deduce that $ink(D) \leq \gamma A$.	Then we use the evaluation of $ink(D)$ in equation~\ref{eq:inka1}. \qed
\end{proof}


\section{Drawing optimization \label{sec:opti}}
From the INKA-area inequality, we can compute the bounds for selected drawing factors.

\subsection{Disk radius $r$}
Given a fixed maximum density, one can find a hard upper bound of node radius. 

The radius of a disk is given by $r \leq \sqrt{ \gamma A / (n \pi)}$. The equality holds when edge width $w$ is 0. 
For example, with a 10x10 drawing ($A$ = 100) and drawing density $\gamma$ = 0.5, drawing 4 nodes would require the radius of each node no greater than $\sqrt{ .5 \times 100 / (4 \pi)}$ = 1.99. 

In general, the inequality can be rewritten as:
$ink(D) = \pi n (r - \frac{mw}{\pi n})^2 + wL -w^2.cr(D) - \frac{m^2 w^2}{\pi n} \leq \gamma A$. Hence, $\pi n (r - \frac{mw}{\pi n})^2 \leq \gamma A - wL + w^2.cr(D) + \frac{m^2 w^2}{\pi n}$. So the radius is bounded by:
$max(0, -\sqrt{\frac{B}{\pi n}} + \frac{wm}{\pi n}) \leq r \leq \sqrt{\frac{B}{\pi n}} + \frac{wm}{\pi n},$ where $B = \gamma A - wL + w^2.cr(D) + \frac{m^2 w^2}{\pi n}$.

In general, a larger radius $r$ requires more ink and thus it leads less ink-effective drawing. 
Figures~\ref{fig:examples}(a)-(b) depict drawings of the same graph using different disk radius.

\begin{figure*}\centering
	\subfloat[Large nodes]{
		\includegraphics[width=.2\textwidth]{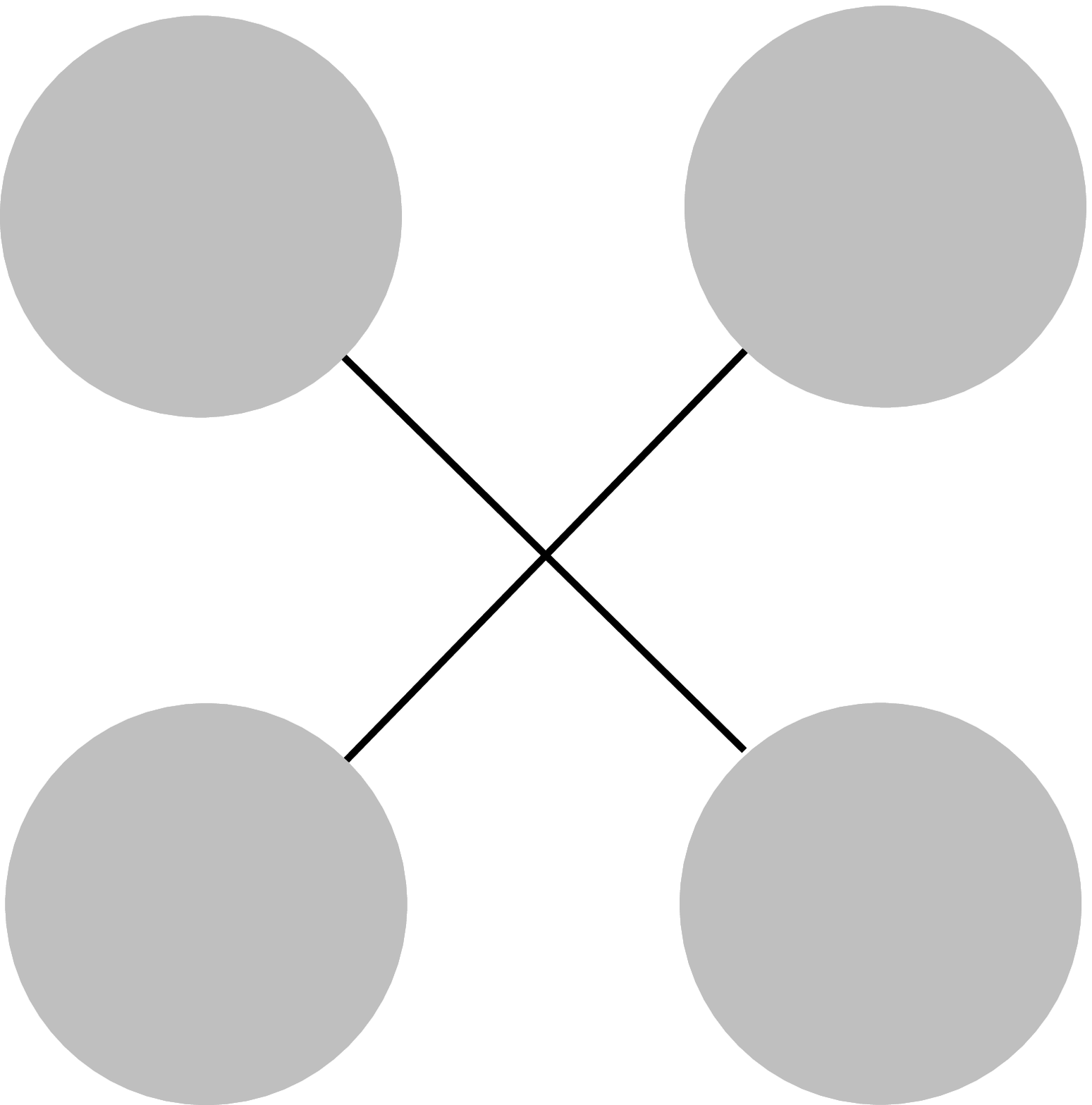}}
	\hspace{5pt}
	\subfloat[Small nodes]{
		\includegraphics[width=.08\textwidth]{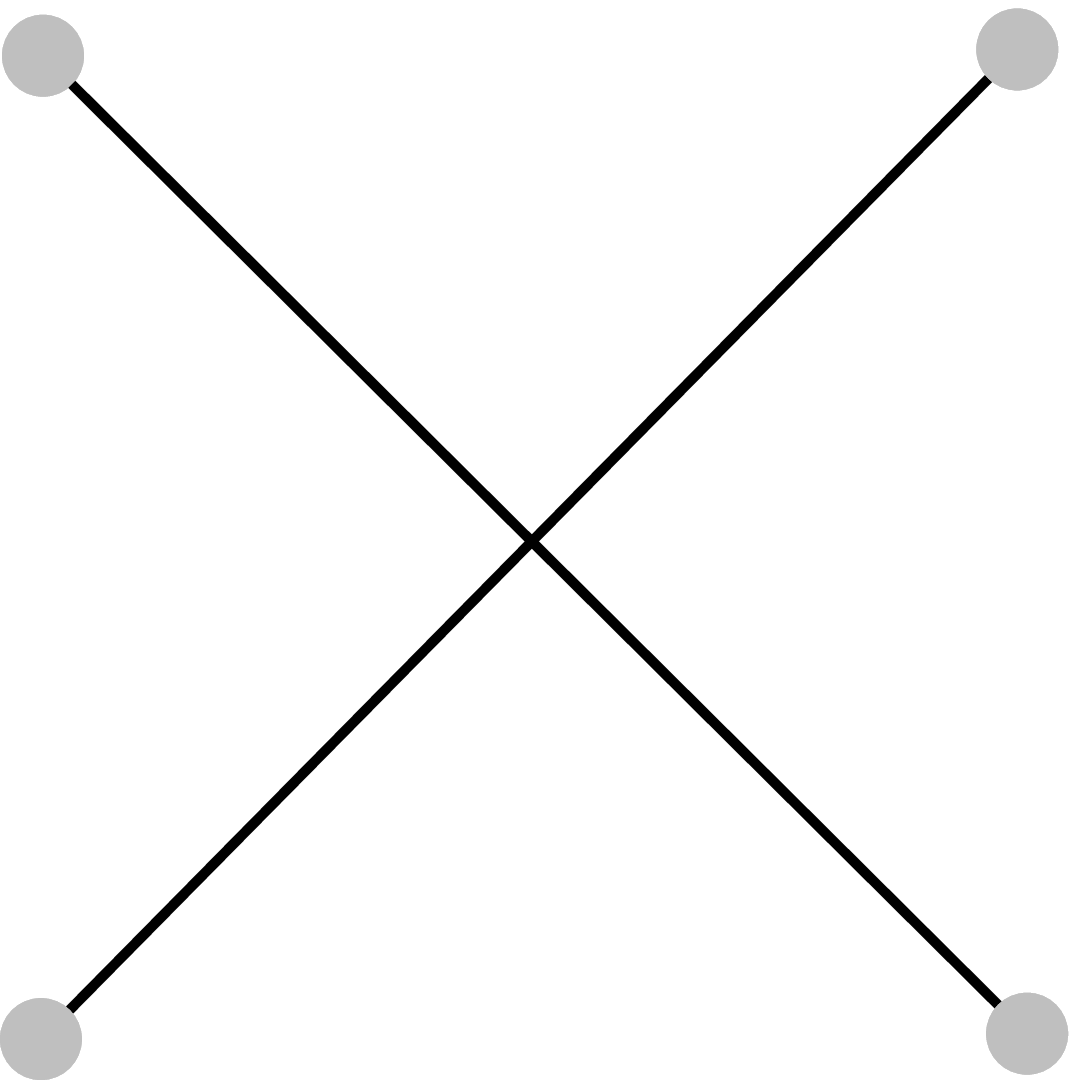}}
	\hspace{.5cm}
	\subfloat[Thin edges ($t$=1)]{
		\includegraphics[width=.15\textwidth]{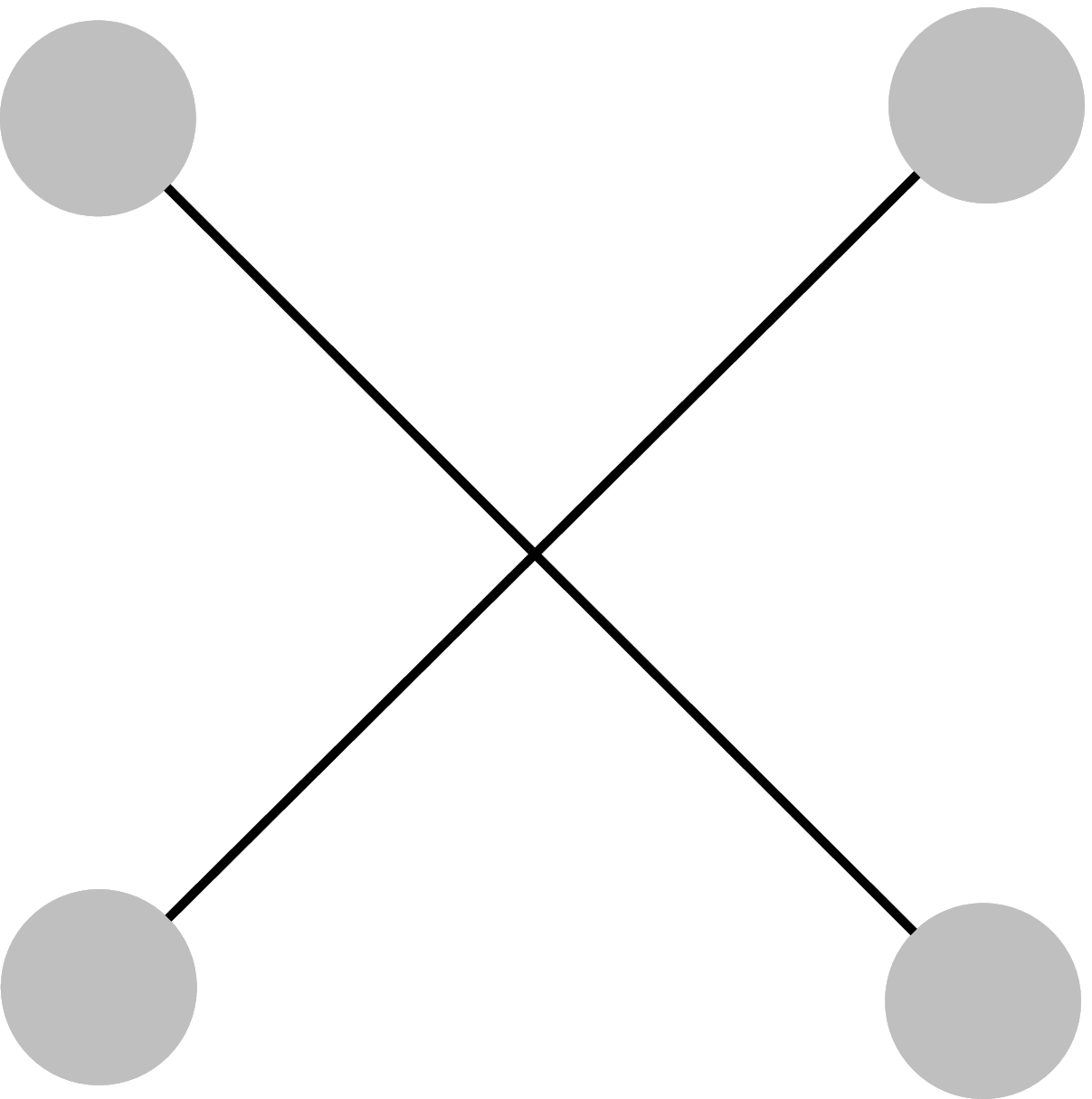}}
	\hspace{5pt}
	\subfloat[Thick edges ($t$ =3)]{
		\includegraphics[width=.15\textwidth]{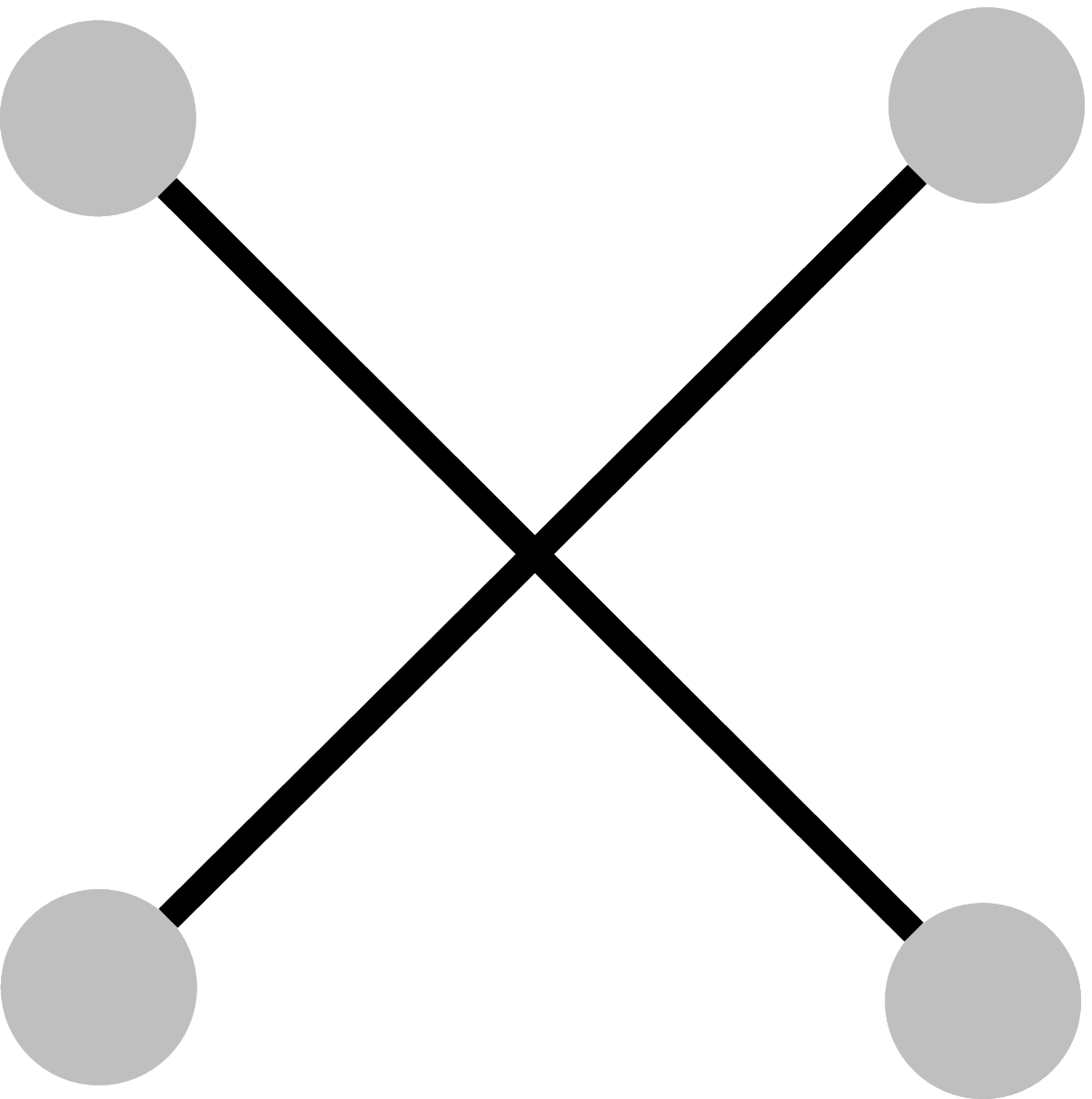}}
	\hspace{.5cm}
	\subfloat[Small]{
		\includegraphics[width=.07\textwidth]{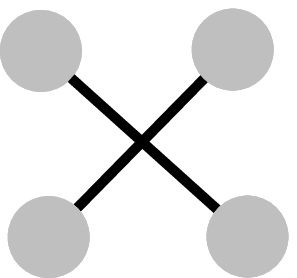}}
	\hspace{5pt}
	\subfloat[Large]{
		\includegraphics[width=.15\textwidth]{./figures/nodesize1.eps}}
	\caption{Drawing factors comparison\label{fig:examples}}
\end{figure*}

\subsection{Edge width $w$}
Similarly, we can find bounds for edge width.
For a fixed radius $r$, the edge width $w$ satisfies:
$0 \leq n \pi r^2 + w (L - 2mr) - w^2. cr(D) \leq \gamma A$. Refactoring gives us $$0 \leq (w - \frac{L-2mr}{cr(D)})^2  + \frac{n\pi r^2}{cr(D)} - (\frac{L-2mr}{cr(D)})^2 \leq \frac{\gamma A}{cr(D)}.$$ Thus, it gives
$ (\frac{L-2mr}{cr(D)})^2 - \frac{n\pi r^2}{cr(D)}  \leq (w - \frac{L-2mr}{cr(D)})^2   \leq \frac{\gamma A}{cr(D)} (\frac{L-2mr}{cr(D)})^2 - \frac{n\pi r^2}{cr(D)}  .$ For example, when radius $r$ is 0, it gives us $w \leq L/cr(D)$.

From the INKA-total equation, the larger the edge width $w$, the more ink is used. Therefore, it leads larger edge width results in lower ink-effectiveness. 
Figures~\ref{fig:examples}(c)-(d) depict drawings of the same graph using different edge width.

\begin{figure*}\centering
	\subfloat[Parallel edges]{
		\includegraphics[width=.15\textwidth]{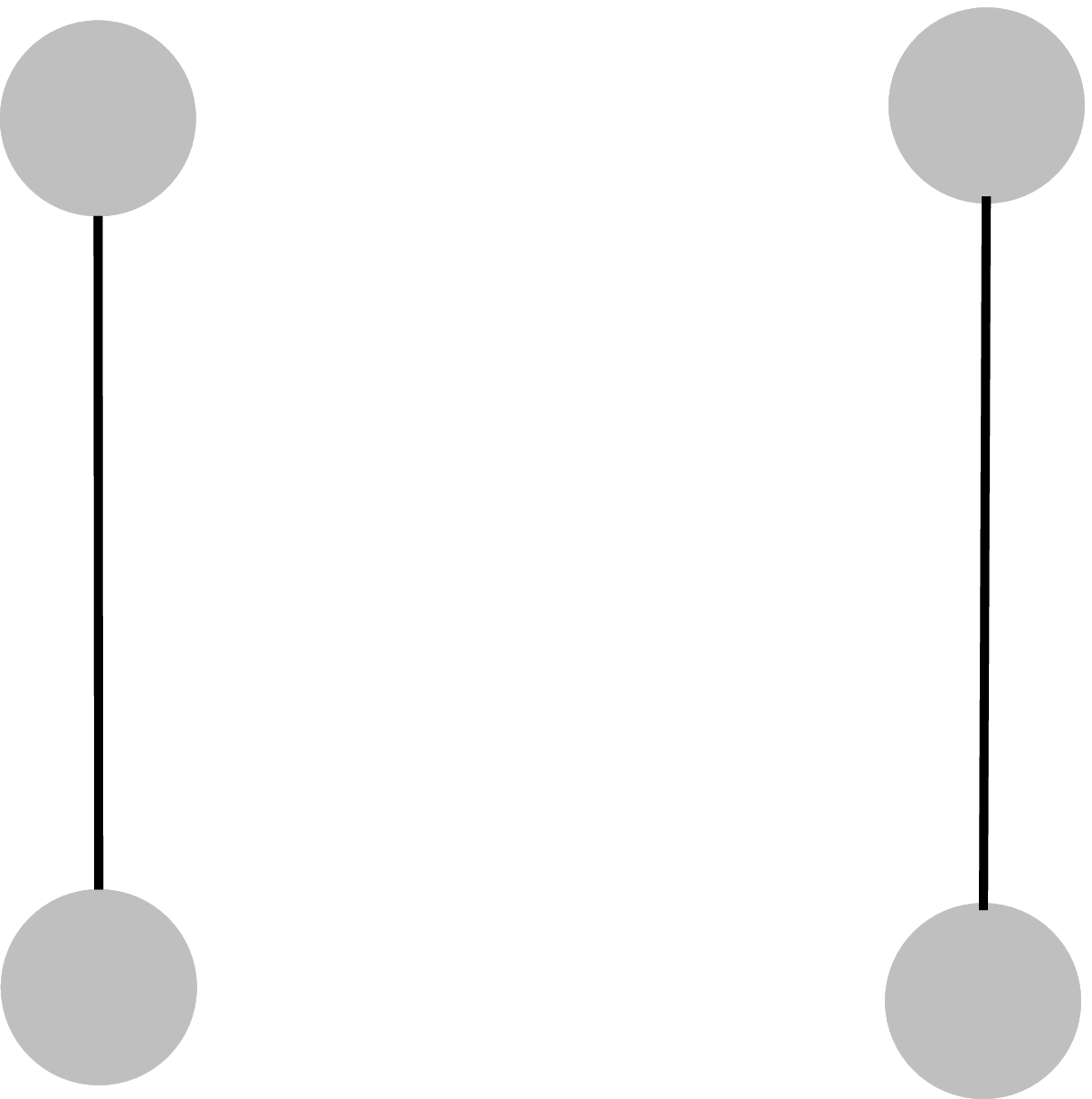}}
	\hspace{1.5cm}
	\subfloat[Crossing edges]{
		\includegraphics[width=.15\textwidth]{./figures/cross.eps}}
	\hspace{1.5cm}
	\subfloat[Crossing edges]{
		\includegraphics[width=.15\textwidth]{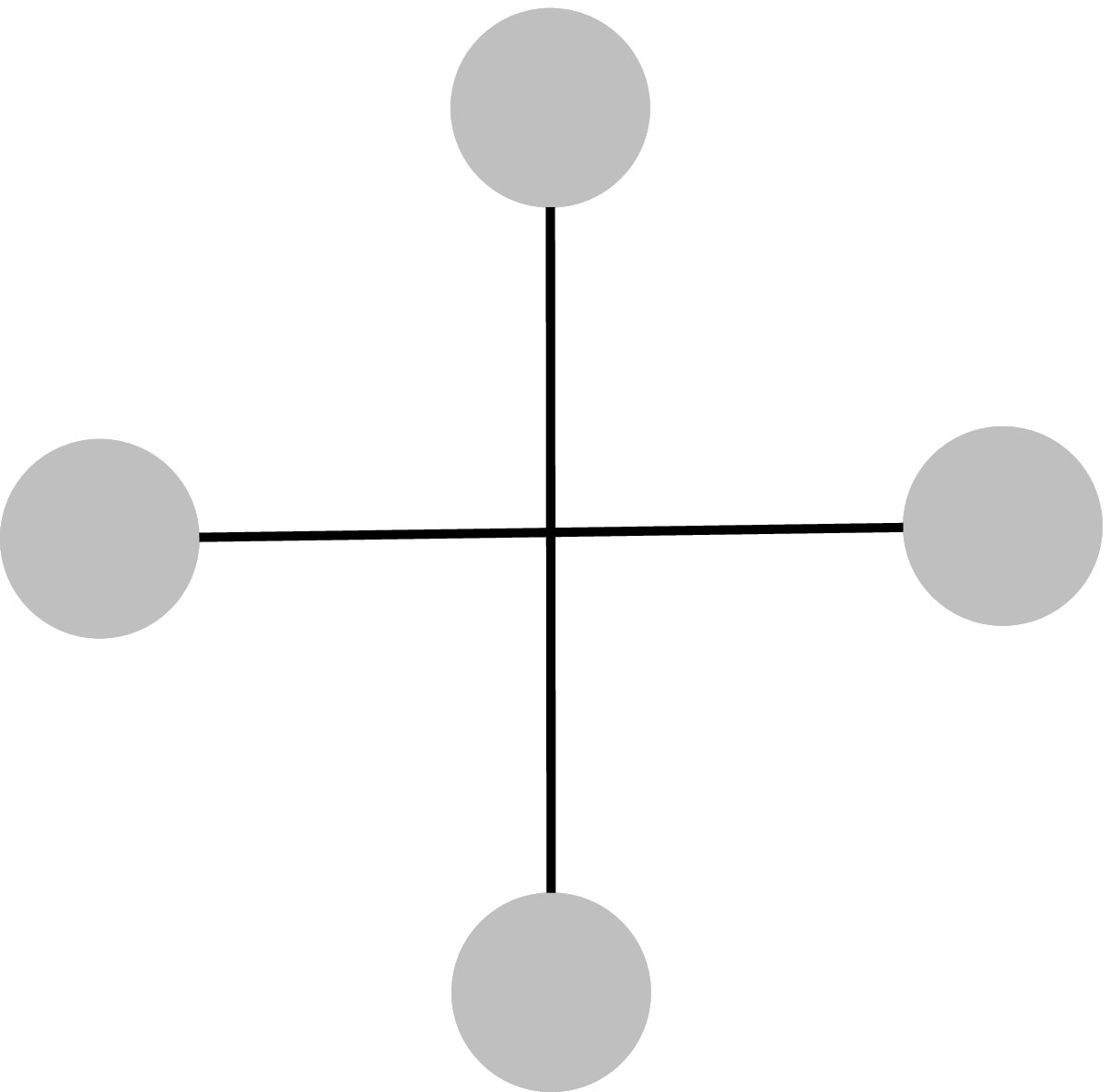}}
	\caption{Ink Comparison of drawings\label{fig:paraVSCross}}
\end{figure*}

\subsection{Edge crossing $cr(D)$}
Now let us consider edge crossings. Figure~\ref{fig:paraVSCross} depicts three different drawings of two pairs of nodes ($u$, $v$), ($w$, $z$) connected by two straight-line edges. Edges have thickness of 0.1.

For Figure~\ref{fig:paraVSCross}(a) and (b), the nodes of radius of 1 are located at the square (0,0), (0,10), (10,0), (10,10).
Two edges (u,v) and (w,z) are non-crossing in in Figure~\ref{fig:paraVSCross}(a); the total ink used is $ink(D) = 4 \pi r^2 + w(2 * 10 - 2 * 2 * r) = 4\pi + 0.1(20 - 4) = 14.16$. In Figure~\ref{fig:paraVSCross}(b), the edges (u,v) and (w,z) are crossing. The ink used is $ink(D) = 4 \pi r^2 + w(2 * 10 \sqrt{2} - 2 * 2 * r) - w^2 * 1= 4\pi + 0.1 * (20 \sqrt{2}- 4) - 0.1^2 = 14.98$.
Thus, in this example drawing with parallel edges is better, i.e., less ink and no crossing. This is an expected result.

Figure~\ref{fig:paraVSCross}(c) differs from Figure~\ref{fig:paraVSCross}(a) by an edge crossing. The ink $ink(D) = 4 \pi r^2 + w(2 * 10 - 2 * 2 * r) - w^2 * 1 = 4 \pi + 0.1*(20 - 4) - 0.1^2 = 14.15$. 
Thus, compared with Figure~\ref{fig:paraVSCross}(b), Figure~\ref{fig:paraVSCross}(c) has the same number of crossings (1), but is more ink-effective.
Overall, Figure~\ref{fig:paraVSCross}(c) is the most ink-effective of the three, despite the crossing.

In the INKA-total equation~\ref{eq:inka1}, when $r$, $t$, $n$ and $m$ are constants, the total of ink $ink(D)$ is proportional to $L - cr(D)$. Thus, to be more ink effective, the total edge length is small (e.g., reducing readability) and the number of edge crossings may be high (e.g., increasing ambiguity).

\subsection{Amount of ink $ink(D)$}

The equation for total ink in the drawing can be written as
$ink(D) = n \pi r^2 + w(L - 2mr) - w^2.cr(D) = n \pi r^2 - 2mwr + wL - w^2 . cr(D)$. That is, $$ink(D) = \pi n (r - \frac{mw}{\pi n})^2 + wL -w^2.cr(D) - \frac{m^2 w^2}{\pi n}
.$$ Hence, the minimum ink amount $ink(D)$ is $${\tt min}\ ink(D) = wL -w^2.cr(D) - \frac{m^2 w^2}{\pi n}.$$ The minimum value of $ink(D)$ is achieved when the radius $r = \frac{wd}{\pi}$, where $d$ is the graph density $m/n$.

\subsection{Remarks}
Here, we give some remarks of the \INKA~model in the context of scaling and zooming.

\subsubsection{Scaling}
A common method to improve readability of a drawing is to scale the node positions. This is useful in many cases, especially when the nodes are placed too close to one another.

Consider a simple scaling that simply scales (up) node positions without changing neither node sizes nor edge thickness. 
Let $s$ be the scale factor ($s$ is greater than 1).

Each edge length $l$ is scaled to $s^2 l$. The total edge length $L'$ becomes $s^2 L$. The total area required becomes $s^2 A$. So the total ink used in the scaled drawing is given by:
$$
ink(D') = n \pi r^2 + w (s^2 L - 2mr) - w^2 . cr(D).
$$
The ink difference between the scaled drawing and the original drawing becomes:
$$ink(D') - ink(D) = w(s^2 L - L) = w (s^2 -1) L. $$ This gives the following lemma:
\begin{lemma}
	The amount of ink difference between a scaled drawing and the original one is proportional to the total edge length and edge width in the original drawing.
\end{lemma}

\subsubsection{Zooming}
Now let us consider zooming. Zooming is different from scaling, in that scaling only scales node positions, whereas zooming scales the node positions, node sizes and edge thickness. 

For a zooming of $s$ times, all nodes and edges are enlarged by $s^2$ times. 

The total ink becomes:
$$ink(D') = s^2 . ink(D) = s^2 n \pi r^2 + s^2 w (L  - 2mr) - s^2 w^2.cr(D)).$$ 

The amount of ink difference is $$ink(D') - ink(D) = (s^2-1) . ink(D).$$ Intuitively, we can deduce that zoom requires more ink than scaling.

Furthermore, $ink(D') = s^2 ink(D) \leq \gamma s^2 A = \gamma A'$. Thus, this gives the following lemma:

\begin{lemma}
	If the original drawing satisfies a drawing ratio $\gamma$, the zoomed drawing also satisfies the drawing ratio.
\end{lemma}

\section{Layout examples \label{sec:studies}}
We now evaluate the \INKA~model using three common graph drawing approaches, which are chosen to demonstrate the usefulness of \INKA. We use our \INKA~model to show the underlying intuitions of these approaches.

\subsection{Example 1: Planar graph drawing}
Planar graph drawing has been extensively studied~\cite{DBLP:reference/crc/2013gd}.
Planar drawings have no edge-edge crossings; that is, $cr(D)$ is 0.
Thus, the total ink used becomes:
$$
ink(D) = n \pi r^2 + w (L - 2mr).
$$
This gives us several interesting results.

First, the total ink is proportional to the total edge length $L$. This is because the values of $r$, $w$, $m$ and $n$ are often considered as constants in planar graph drawing algorithms.  Thus, improving ink effectiveness for planar graph drawing is equivalent to minimizing the total edge length $L$.

\begin{lemma}
	The ink effectiveness for planar graph drawing is equivalent to the minimization of total edge length $L$.
\end{lemma}

This is a remarkable result. Many force-directed algorithms that draw planar or near-planar graphs can indirectly achieve the minimum total edge length.

Second, given a planar drawing of a graph $G$, the edge width is constrained by: $$w \leq (\gamma A - n \pi r^2)/(L-2mr).$$ When $r$ is 0, then $w$ $\leq$ $\frac{\gamma A}{L}$.

Third, another interesting result is that $m$ is bounded by 3$n$ - 6 for planar graphs. For maximal planar graph, the INKA-area inequality gives 
$$ink(D) = n \pi r^2 + w [L - 2(3n-6)r] \leq \gamma A$$.
Thus, the total edge length is bounded by: $$L \leq \frac{1}{w} [\gamma A-12rw - n ( \pi r^2 - 6wr) ].$$
This gives the maximum edge length $L_{max}$ is $1/w [\gamma A-12rw - n (\pi r^2 - 6wr) ]$ for maximal planar graphs.
For example, when $r$=1 and $w$=1, then $L_{max} \approx \gamma A - 12 + 2.85n$.

\subsection{Example 2: Equal-edge-length drawing}
This section presents the INKA model for a special class of graph drawing, which all edges have the same length.
Several force-directed algorithms implicitly optimize for equal (fixed) edge length, via so-called 'preferred edge length'. Examples include, for example, the work of ~\cite{Kamada89analgorithm,FruRei91,Frick:GEM1994,Davidson:1996}. 


Now, let consider a drawing in which all edges have the same length of $l$. The total edge length becomes $L= ml$ for $m$ edges. Our INKA-area inequality in Equation~\ref{eq:drawingratio} gives:
$$
0 \leq ink(D) = n \pi r^2 + w (ml - 2mr) - w^2.cr(D) \leq \gamma A.
$$

One can determine some bounds. First, when $r$ is 0, the inequality gives the neccessity condition for $l$: $w^2 . cr(D) \leq wml \leq \gamma A  + w^2. cr(D)$. Thus, the length $l$ is bounded by: 
$\frac{w.cr(D)}{m} \leq l \leq  \frac{\gamma A}{wm} + \frac{w.cr(D)}{m}.$
Second, the number of crossing in an equal-length drawing is bounded by $cr(D) \leq \frac{ml}{w}.$ Often the values $l$ and $w$ are fixed, the number of crossings is satisfied:
\begin{lemma}
	The number of crossings $cr(D)$ in an equal-length drawing is bounded by the number of edges.	
\end{lemma}

These conditions must be met in order to achieve a drawing in which all edges have equal length.

\subsection{Example 3: Partial edge drawing}
Now we diskuss our \INKA~model for \emph{partial edge drawing}~\cite{bruckdorfer2012progress}. Partial edge drawing avoids crossings by dropping the middle part of edges and showing only the remaining edge parts.

Let $p$ is the partial edge ratio, which specifies the proportion of edges are still displayed; for example, $p$ can be 0.1, 0.2, 0.5 to 1. When $p$ is 1, this is equal to the normal drawing (full edges).

From INKA model, the total ink for a partial drawing $D'$ of $D$ is approximated by:
$$
ink(D') = n \pi r^2 + w (pL - 2mr) - w^2.cr(D'),
$$
where $D'$ is the partial drawing with the same disk radius $r$ and edge width $w$.

For partial drawings, it is believed that the amount of ink $ink(D')$ of the partial drawing $D'$ is smaller than the amount of ink in $ink(D)$. Also the number of crossings $cr(D')$ in a partial drawing is often expected to be smaller than the number of crossing  $cr(D)$ in the normal drawing.

In fact, our INKA model gives the ink difference $ink(D') - ink(D) =  wpL - wL - w^2.cr(D') + w^2.cr(D) = wL (p-1) + w^2 (cr(D) - cr(D'))$. Thus, $ink(D') \leq ink(D)$ only if $wL (p-1) + w^2 (cr(D) - cr(D')) \leq 0$. This gives the necessity condition of $cr(D) - cr(D') \leq (1-p) L/w $. Remarkably, in realistic settings the total edge length $L \gg w$ and thus the necessity condition always holds. This is the reason for $ink(D') \leq ink(D)$.

Another interesting result of using INKA model is that one can work out the bounds for $cr(D')$. For example, when $r$ is 0, the inequality $ink(D') \leq ink(D) \leq \gamma A$ gives the bounds for the number of crossings $cr(D')$:
$pL/w - \gamma A / w^2 \leq cr(D') \leq p L / w.$
Besides, 
$w.cr(D') \leq p L \leq \gamma A / w + w.cr(D')$.

In general, partial edge drawing reduces ambiguity and edge readability at the same time.

\section{Evaluation \label{sec:eval}}

This section presents our evaluation of the \INKA~model using real-world graphs. The aim is to have an approximation of the ink used for various graphs using different layout algorithms.

\subsection{Data sets}
Here, we use several ``benchmark'' data sets, which are from the Hachul library, Walshaw's Graph Partitioning Archive, the sparse matrices collection~\cite{Davis:2011} and the network repository~\cite{nr-aaai15}. These data sets include commonplace types of graphs: grid-like graphs and scale-free graphs. Table~\ref{table:data} shows the graphs used in our experiment.

\begin{table}[t]\centering
	\renewcommand{\arraystretch}{1.5}
	\setlength{\tabcolsep}{20pt}
	\caption{\label{table:data}Data sets}
	\begin{tabular}{|l|c|c|c|c|}
		\hline
		graph &	$|V|$	& $|E|$ \\
		\hline
		can\_144 & 144 &	576 \\
		G\_2 	& 4970	& 7400 \\
		G\_3	& 2851	& 15093 \\	
		G\_4	& 2075	& 4769 \\
		G\_15 	& 1785	& 20459 \\	
		mm\_0	& 3296	& 6432\\	
		nasa1824	& 1824	& 18692  \\
		yeastppi & 2361	& 7182 \\		
		\hline
	\end{tabular}	
\end{table}

\subsection{Design}
We compare the amount of ink computed by the \INKA~model for different layouts of the data sets. For layout, we used on the standard \emph{FM3} layout~\cite{Hachul:2004} and its variants, which are implemented in OGDF~\cite{MarkusChimani2012}. The FM3 variants include Multi-level Fast (Fast), Multi-level Nice (Nice) and Multi-level NoTwist (NoTwist).

\subsection{Results}
We computed the number of crossings and the total edge length for each resulting layout.
Figure~\ref{fig:stats} shows the statistics of the graph layout results. The y-axis shows a logarithmic scale. As shown in the figure, the number of crossings varies a lot between the graphs and the layouts. The total edge length appears to be proportional to the number of edges $M$ and the number of vertices $N$. 

\begin{figure*}\centering
	\includegraphics[width=0.95\linewidth]{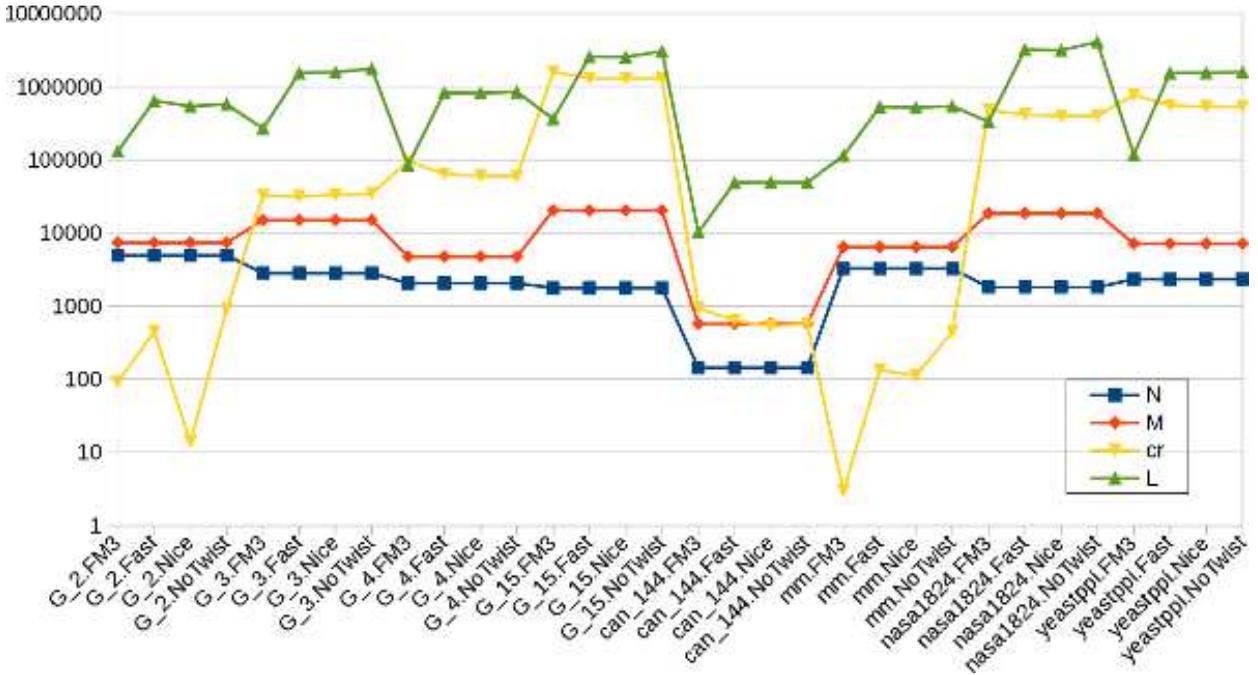}
	\caption{\label{fig:stats}Statistics of graph layouts.}
\end{figure*}
\begin{figure*}
	\centering
	\includegraphics[width=.95\linewidth]{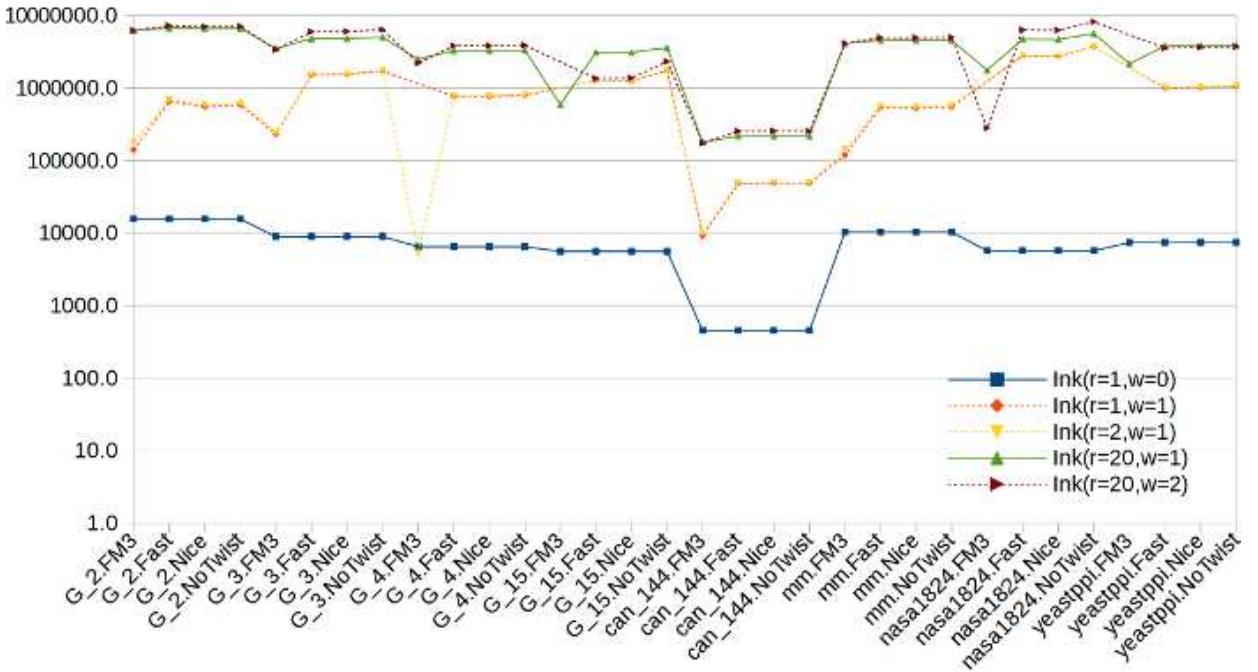}
	\caption{\label{fig:ink}Ink approximation of all data sets using INKA-total equation.}
\end{figure*}

We compared the amount of ink of the same graph layouts using different settings of vertex radius and edge width. The first setting $r$=1 and $w=0$ is used as the base.
This means the drawing only shows the vertices, but not the edges. The other settings are ($r$=1, $w$=1), ($r$=2, $w$=1), ($r$=20, $w$=1) and ($r$=20, $w$=2) to test the variations of node radius and edge width.

Figure~\ref{fig:ink} shows the estimated amount of ink using \INKA~for all of the graphs. The y-axis shows a log scale. There are several interesting results. First, the figure depicts that the larger the values of $r$ and $w$, the more ink is used. For example, the blue line ($r$ = 1 and $w$ = 0) is as the base line for all the other lines. Drawing only vertices require less ink than the other settings.
Second, when the radius is changed slightly (for example, from 1 to 2) the total amount of ink is slightly changed (see the red and the yellow lines).
Third, when the edge width slightly varies, the total amount of ink also changes slightly (see the green and plum color lines). Fourth, the amount of ink may decrease with larger node radius or with
large edge width.
Last but not least, FMMM layouts appear to use less ink than the other layouts, for all data sets.

\section{Discussions \label{sec:diskussion}}

\subsection{Data-ink in graph visualization}

An important criteria in visualization is the data-ink ratio by Tufte~\cite{tufte1983visual}.
In the perspective of graph visualization, the data-ink (or accordingly we call \emph{graph-ink}) is the non-erasable ink that presents nodes and edges. Removing the data-ink from the drawing would cause a missing of node(s) or edge(s).  The graph-ink ratio is the proportion of the graph-ink compared to the total amount of ink (or pixels) used in the drawing.

Maximizing the graph-ink ratio is equivalent to minimizing the total graph-ink used to present the graph.
The less ink used in $D$ to draw $G$, the better. In fact, ink minimization has been studied in graph layout algorithms \cite{sugiyama1994methods,north2001online,gansner2006improved,gansner2011multilevel}.

\subsection{Drawing factor relationship}
From our \INKA~model, we summarize the relationship between drawing factors.
Figure~\ref{fig:inkdiagram} depicts a diagram that summarizes the relationship between the common drawing factors.

\begin{figure*}\centering
	\includegraphics[width=.7\textwidth]{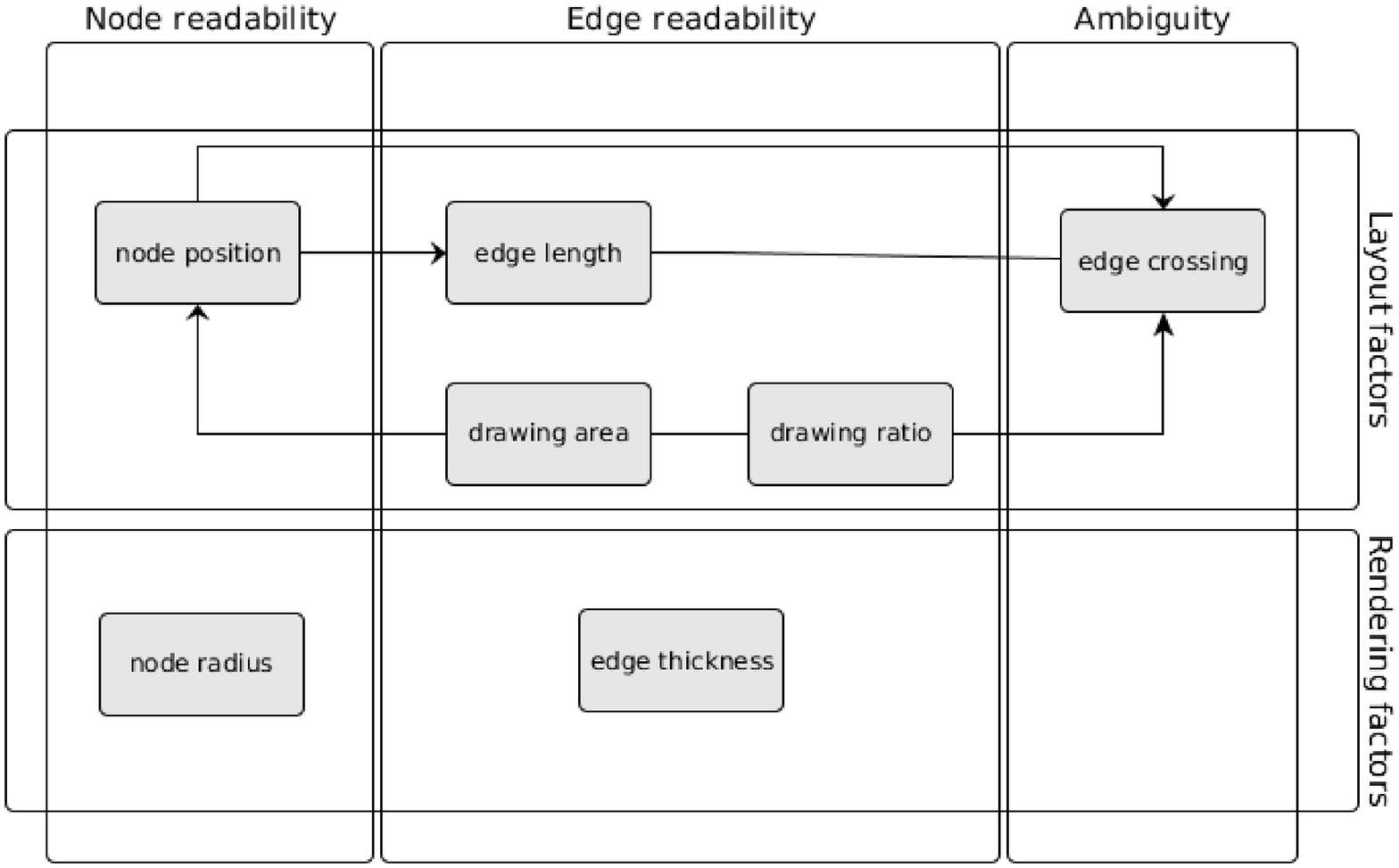}
	\caption{Diagram of Drawing factors\label{fig:inkdiagram}}
\end{figure*}

\subsubsection{Ink-effectiveness vs. Clarity}
In contrast, the ink expression in the Equation~\ref{eq:inka} can be interpreted as follows:
\begin{equation}\label{eq:readabityVSink}
ink(D) = clarity(V) + clarity(E) - overlap,
\end{equation}
where the clarity of nodes $clarity(V)$ is defined as the total amount of ink used for drawing all nodes;
the clarity of edges $clarity(E)$ is equal to the total amount of ink used for drawing all edges; and
the $overlap$ is account to the ambiguity.


From Equation~\ref{eq:readabityVSink}, one can deduce that the more ink-effective (or less ink used), the less readability and also the more ambiguity the drawing becomes. The larger the $overlap$, the less faithful~\cite{NguyenEH12a} the drawing becomes (e.g., more overlap makes it harder to derive the original graph from the drawing). Thus, this is the side effect of minimizing the total ink for graph drawing.

Often, increasing node size and edge thickness may improve clarity of the drawing.
Let us consider how an increase of two factors affect the ink-effectiveness.

Consider a graph $G$ and a fixed layout $D$ of $G$. Let us consider two drawings $D$ and $D'$ with variations in node size and edge thickness. Let $ink(D)$ and $ink(D')$ be the ink measures of $D$ and $D'$, respectively.

For drawings with the same $r$ (i.e. $r$=$r'$) but different edge widths  ($w$ and $w'$), the total ink difference between $D$ and $D'$ is given by:
$$ink(D') - ink(D) = (w' - w)[ L - 2mr - (w + w'). cr(D)].$$
Thus, thicker edges may or may not improve ink effectiveness; it depends on the number of crossings existing in the drawing. Interestingly, one could find 
that for a pair of widths $w$ and $w'$, the total ink is the same between two drawings if $L - 2mr = (w + w'). cr(D)$.

For drawings with same thickness $t$ (i.e., $t$= $t'$) but different node radius ($r$ and $r'$), the total ink difference  between $D$ and $D'$ is given by:
$$ink(D') - ink(D) = n \pi (r'^2 - r^2).$$ This implies that larger node size improves  (node) readability quadratically while reducing ink effectiveness.


\subsection{Limitations}
Our \INKA~model gives a fast way to approximate the total amount of ink used for a layout of a graph before the rendering has been achieved. The approximation of $overlap$ described in Section~\ref{sec:inka} is somewhat simplistic. However, in practice, the edge width $w$ is often set to a small value (so as to see the lines); thus, the approximation is still a good estimation of the amount of ink.

\section{Conclusion and Future work \label{sec:conclusion}}

In this paper, we have introduced a new \INKA~model and have applied the model to analyse common drawing factors used in graph drawing. The relationship between the most common drawing factors is encapsulated in the INKA-total equation (Eq. \ref{eq:inka1}) and INKA-area inequality in Section~\ref{sec:inka}. (Eq.\ref{eq:drawingratio}). The common drawing factors include edge crossing, total edge length, drawing area, node radius and edge width.

Overall, the \INKA~model gives a useful foundation to estimate the feasibility of a layout design for certain values of drawing factors. The new model also gives a way to  approximate the total ink used in a drawing.  
We have demonstrated several use cases of our \INKA~model. We also have presented our evaluation of \INKA~for real-world data sets using different layouts.


Our examples and experimental results of the \INKA~model have motivated several directions for future work. First, one can integrate \INKA~model into graph layout algorithms to better lay out and render graphs. Second, to be ink effective, the total edge length $L$ should be small while the number of crossings $cr(D)$ may be large. But good drawings should balance between ink-effectiveness while keeping the number of crossings (ambiguity) small. This can be formulated as a minimization problem of $Q(D) = \alpha. ink(D) + \beta. cr(D)$, for some non-negative numbers $\alpha$ and $\beta$. Third, it would be interesting to extend \INKA~model to model and understand about the relationships among other layout factors (such as minimum edge length, symmetry, node distribution, orthogonality and crossing angles) and other rendering factors (such as color and transparency).



\ifCLASSOPTIONcompsoc
  \section*{Acknowledgments}
\else
  \section*{Acknowledgment}
\fi

The authors would like to thank anonymous reviewers for very helpful feedbacks.





\bibliographystyle{abbrv}
\bibliography{inka,force,metric}

\vspace{-1cm}



\clearpage
\enlargethispage{-5in}

\appendices

\end{document}